\documentclass{article}
\usepackage{amsthm,authblk,fullpage,amssymb}

  \usepackage{amsmath,amsfonts}
  \usepackage{tikz}
  \usepackage[T1]{fontenc}
  \usepackage[utf8]{inputenc}
  \usepackage{comment}
  \usepackage{xspace}
  \usepackage{enumerate}
  \usepackage{listings}
  \usepackage{todonotes}
  \usepackage{array}
  \usepackage{enumitem}
  \usepackage[noend, noline, ruled]{algorithm2e}
  
  \newtheorem{observation}{Observation}
  \newtheorem{fact}{Fact}
  \newtheorem{lemma}{Lemma}
  \newtheorem{corollary}{Corollary}
  \newtheorem{theorem}{Theorem}

  \usepackage{thmtools}
  \usepackage{thm-restate}
  \usepackage[capitalise]{cleveref}
  
  \usetikzlibrary{decorations.pathreplacing,calc,snakes,patterns}

  \crefname{fact}{Fact}{Facts}
  \crefname{observation}{Observation}{Observations}
  \SetKwInOut{KwRequire}{Require}
  \SetKwComment{Comment}{$\triangleright$\ }{}

  \newcommand{\ceil}[1]{\left\lceil #1 \right\rceil}
  \newcommand{\Oh}{\mathcal{O}}
  \newcommand{\sub}{\subseteq}

  \newcommand{\assign}{\leftarrow}

  \newcommand{\LCE}{\mathrm{LCE}}
  \newcommand{\lcp}{\mathrm{lcp}}
  \newcommand{\nc}{\mathrm{nc}}
  \newcommand{\LTree}{\mathit{LTree}}
  \newcommand{\idx}{r}
  
  \newcommand{\LevelLCE}[1]{\LCE^{(#1)}}
  \newcommand{\LimitedLCE}{\mathrm{Limited\mbox{-}LCE}}

  \newcommand{\leaftext}[1]{$\mathtt{#1}$}

  \usetikzlibrary{trees, decorations.pathmorphing, shapes, arrows, positioning}
  \usepackage{forest}

  \tikzstyle{internal} = [circle, minimum width=3pt,fill, inner sep=1pt, label={[black,below]0:#1}]
  \tikzstyle{leaf} = [rectangle, minimum width=6pt,fill, inner sep=1pt, label={[black,below]0:#1}]
  \tikzstyle{vertex} = [rectangle, minimum width=1pt,fill, inner sep=1pt]
  \tikzstyle{zig} = [draw,thick, decorate,
  decoration={zigzag,amplitude=2pt,segment length=2mm,pre=lineto,pre length=2pt, post=lineto,post length=4pt}, inner sep=5pt]
  \tikzstyle{arr} = [draw,-,thick, inner sep=5pt]

  \title{Near-Optimal Computation of Runs over General Alphabet via Non-Crossing LCE Queries}

  \author{
    Maxime Crochemore
  }
  \author{
    Costas S.\ Iliopoulos
  }
  \author{
    Tomasz Kociumaka\thanks{Supported by Polish budget funds for science in 2013-2017 as a research project under the 'Diamond Grant' program.}
  }
  \author{
    Ritu Kundu
  }
  \author{
    Solon P.\ Pissis
  }
  \author{
    Jakub Radoszewski\thanks{The author is a Newton International Fellow.}
  }
  \author{
    Wojciech Rytter\thanks{Supported by the Polish National Science Center, grant no 2014/13/B/ST6/00770.}
  }
  \author{
    Tomasz Wale\'n$^{\star\,\star\,\star}$
  }

  \affil{
    Department of Informatics, King's College London, London, UK\\
    \texttt{[maxime.crochemore,costas.iliopoulos,ritu.kundu,solon.pissis]@kcl.ac.uk}
  }

  \affil{
    Faculty~of Mathematics, Informatics and Mechanics,\\
    University of Warsaw, Warsaw, Poland\\
    \texttt{[kociumaka,jrad,rytter,walen]@mimuw.edu.pl}
  }

  \date{}

\begin{document}
  \maketitle
\begin{abstract}
  Longest common extension queries (LCE queries) and runs are ubiquitous in algorithmic stringology.
  Linear-time algorithms computing runs and preprocessing for constant-time LCE queries have been known for over a decade.
  However, these algorithms assume a linearly-sortable integer alphabet.
  A recent breakthrough paper by Bannai et.\ al.\ (SODA 2015) showed a link between the two notions: all the runs in a string
  can be computed via a linear number of LCE queries.
  The first to consider these problems over a general ordered alphabet was Kosolobov (\emph{Inf.\ Process.\ Lett.}, 2016),
  who presented an $\Oh(n (\log n)^{2/3})$-time algorithm for answering $\Oh(n)$ LCE queries.
  This result was improved by Gawrychowski et.\ al.\ (accepted to CPM 2016) to $\Oh(n \log \log n)$ time.
  In this work we note a special \emph{non-crossing} property of LCE queries asked in the runs computation.
  We show that any $n$ such non-crossing queries can be answered on-line in $\Oh(n \alpha(n))$ time,
  which yields an $\Oh(n \alpha(n))$-time algorithm for computing runs.
\end{abstract}

  \section{Introduction}
  \emph{Runs} (also called \emph{maximal repetitions}) are a fundamental type of repetitions in a string
  as they represent the structure of all repetitions in a string in a succinct way.
  A run is an inclusion-maximal periodic factor of a string in which the shortest period repeats at least twice.
  A crucial property of runs is that their maximal number in a string of length $n$ is $\Oh(n)$.
  This fact was already observed by Kolpakov and Kucherov \cite{DBLP:conf/focs/KolpakovK99,Kolpakov99}
  who conjectured that this number is actually smaller than $n$, which was known as the runs conjecture.
  Due to the works of several authors \cite{DBLP:conf/mfcs/CrochemoreI07,DBLP:journals/jcss/CrochemoreI08,DBLP:conf/cpm/CrochemoreIT08,DBLP:conf/lata/Giraud08,DBLP:journals/tcs/PuglisiSS08,DBLP:conf/stacs/Rytter06,DBLP:journals/iandc/Rytter07}
  more precise bounds on the number of runs have been obtained, and finally
  in a recent breakthrough paper \cite{DBLP:journals/corr/BannaiIINTT14} Bannai et al. proved the runs conjecture,
  which has since then become the runs theorem
  (even more recently in \cite{DBLP:conf/spire/0001HIL15} the upper bound of $0.957n$ was shown for binary strings).

  Perhaps more important than the combinatorial bounds is the fact that the set of all runs in a string
  can be computed efficiently.
  Namely, in the case of a linearly-sortable alphabet $\Sigma$ (e.g., $\Sigma=\{1,\ldots,\sigma\}$ with $\sigma = n^{\Oh(1)}$)
  a linear-time algorithm based on Lempel-Ziv factorization \cite{DBLP:conf/focs/KolpakovK99,Kolpakov99} was known for a long time.
  In the recent papers of Bannai et al.~\cite{DBLP:journals/corr/BannaiIINTT14,DBLP:conf/soda/BannaiIINTT15} it is
  shown that to compute the set of all runs in a string, it suffices to answer $\Oh(n)$ longest common extension (LCE) queries.
  An LCE query asks, for a pair of suffixes of a string, for the length of their longest common prefix.
  In the case of $\sigma = n^{\Oh(1)}$ such queries can be answered on-line in $\Oh(1)$ time after
  $\Oh(n)$-time preprocessing that consists of computing the suffix array with its inverse, the LCP table and
  a data structure for range minimum queries on the LCP table; see e.g.~\cite{AlgorithmsOnStrings}.
  The algorithms from \cite{DBLP:journals/corr/BannaiIINTT14,DBLP:conf/soda/BannaiIINTT15} use (explicitly
  and implicitly, respectively) an intermediate notion of Lyndon tree (see \cite{DBLP:journals/jct/Barcelo90,DBLP:journals/tcs/HohlwegR03})
  which can, however, also be computed using LCE queries.
  
  Let $T_{\LCE}(n)$ denote the time required to answer on-line $n$ LCE queries in a string.
  In a very recent line of research, Kosolobov \cite{DBLP:journals/ipl/Kosolobov16} showed that, for a general ordered alphabet,
  $T_{\LCE}(n) = \Oh(n (\log n)^{2/3})$, which immediately leads to $\Oh(n (\log n)^{2/3})$-time
  computation of the set of runs in a string.
  In \cite{DBLP:/conf/cpm/GawrychowskiKRW16} a faster, $\Oh(n \log \log n)$-time algorithm for answering $n$ LCE queries has been presented
  which automatically leads to $\Oh(n \log \log n)$-time computation of runs.

  Runs have found a number of algorithmic applications.
  Knowing the set of runs in a string of length $n$ one can compute
  in $\Oh(n)$ time all the local periods and the number of all squares, and also
  in $\Oh(n+T_{\LCE}(n))$ time all distinct squares
  provided that the suffix array of the string is known \cite{DBLP:journals/tcs/CrochemoreIKRRW14}.
  Runs were also used in a recent contribution on efficient answering of internal pattern matching queries
  and their applications \cite{DBLP:conf/soda/KociumakaRRW15}.

  \paragraph{\bf Our Results}
  We observe that the computation of a Lyndon tree of a string and furthermore the computation of all the runs in a string
  can be reduced to answering $\Oh(n)$ LCE queries that are \emph{non-crossing}, i.e.,
  no two queries $\LCE(i,j)$ and $\LCE(i',j')$ are asked with $i<i'<j<j'$ or $i'<i<j'<j$.
  Let $T_{\nc\LCE}(n)$ denote the time required to answer $n$ such queries on-line in a string of length $n$
  over a general ordered alphabet.
  We show that $T_{\nc\LCE}(n) = \Oh(n \alpha(n))$, where $\alpha(n)$ is the inverse Ackermann function.
  As a consequence, we obtain $\Oh(n \alpha(n))$-time algorithms for computing the Lyndon tree,
  the set of all runs, the local periods and the number of all squares in a string over a general ordered alphabet.
  
  Our solution relies on a trade-off between two approaches. The results of \cite{DBLP:/conf/cpm/GawrychowskiKRW16} let us efficiently compute the LCEs if they are short,
  while $\LCE$ queries with similar arguments and a large answer yield structural properties of the string, which we discover and exploit to answer further such queries.

  Our approach for answering non-crossing $\LCE$ queries is described in three sections:
  in \cref{sec:overview} we give an overview of the data structure,
  in \cref{sec:bp} we present the details of the implementation,
  and in \cref{sec:comp} we analyse the complexity of answering the queries.
  The applications including runs computation are detailed in \cref{sec:runs}.
  The appendix contains some supporting examples.

  \section{Preliminaries}
  \paragraph{\bf Strings}
  Let $\Sigma$ be a finite ordered alphabet of size $\sigma$.
  A string $w$ of length $|w|=n$ is a sequence of letters $w[1] \ldots w[n]$ from $\Sigma$.
  By $w[i,j]$ we denote the \emph{factor} of $w$ being a string of the form $w[i] \ldots w[j]$.
  A factor $w[i,j]$ is called proper if $w[i,j] \ne w$.
  A factor is called a \emph{prefix} if $i=1$ and a \emph{suffix} if $j=n$.
  We say that $p$ is a \emph{period} of $w$ if $w[i]=w[i+p]$ for all $i=1,\ldots,n-p$.
  If $p$ is a period of $w$, the prefix $w[1,p]$ is called a \emph{string period} of $w$.
  
  By an interval $[\ell,r]$ we mean the set of integers $\{\ell,\ldots,r\}$.
  If $w$ is a string of length $n$, then an interval $[a,b]$ is called a \emph{run} in $w$
  if $1 \le a < b \le n$, the shortest period $p$ of $w[a,b]$ satisfies $2p \le b-a+1$ and
  none of the factors $w[a-1,b]$ and $w[a,b+1]$ (if it exists) has the period $p$.
  An example of a run is shown in \cref{fig:example-run}.

  \paragraph{\bf Lyndon Words and Trees}
  By $\prec = \prec_0$ we denote the order on $\Sigma$ and by $\prec_1$ we denote
  the reverse order on $\Sigma$.
  We extend each of the orders $\prec_\idx$ for $\idx \in \{0,1\}$ to a lexicographical order on strings over $\Sigma$.
  A string $w$ is called an \emph{$\idx$-Lyndon word} if
  $w \prec_\idx u$ for every non-empty proper suffix $u$ of $w$.
  The \emph{standard factorization} of an $\idx$-Lyndon word $w$ is a pair $(u, v)$ of $\idx$-Lyndon words
  such that $w=uv$ and $v$ is the longest proper suffix of $w$ that is an $\idx$-Lyndon word.

  The \emph{$\idx$-Lyndon tree} of an $\idx$-Lyndon word $w$, denoted as $\LTree_\idx(w)$,
  is a rooted full binary tree defined recursively on $w[1,n]$ as follows:
  \begin{itemize}
    \item $\LTree_\idx(w[i,i])$ consists of a single node labeled with $[i,i]$
    \item if $j-i>1$ and $(u,v)$ is the standard factorization of $w[i,j]$,
      then the root of $\LTree_\idx(w)$ is labeled by $[i,j]$, has left child
      $\LTree_\idx(u)$ and right child $\LTree_\idx(v)$.
  \end{itemize}
  See \cref{fig:example-ltree} for an example.
  We can also define the $\idx$-Lyndon tree of an arbitrary string.
  Let $ \$_0,\$_1$ be special characters smaller than and greater than all the letters from $\Sigma$, respectively.
  We then define $\LTree_\idx(w)$ as $\LTree_\idx(\$_\idx w)$; note that $ \$_\idx w$ is an $\idx$-Lyndon word.

  \paragraph{\bf LCE Queries}
  For two strings $u$ and $v$, by $\lcp(u,v)$ we denote the length of their longest common prefix.
  Let $w$ be a string of length $n$.
  An LCE query $\LCE(i,j)$ computes $\lcp(w[i,n],w[j,n])$.
  An $\ell$-limited LCE query  $\LimitedLCE_{\le \ell}(i,j)$ computes $\min(\LCE(i,j),\ell)$.
  Such queries can be answered efficiently as follows; see Lemma 6.3\ in \cite{DBLP:/conf/cpm/GawrychowskiKRW16}.

  \begin{lemma}[\cite{DBLP:/conf/cpm/GawrychowskiKRW16}]\label{lem:CPM}
    A sequence of $q$ queries $\LimitedLCE_{\le \ell_p}(i_p,j_p)$ can be answered on-line
    in $\Oh((n + \sum_{p=1}^q \log \ell_p)\alpha(n))$ time over a general ordered alphabet.
  \end{lemma}
  
  The following observation shows a relation between LCE queries and periods in a string
  that we use in our data structure; for an illustration see \cref{fig:fajny}.

\renewcommand{\AA}{\mathrm{A}}
\newcommand{\BB}{\mathrm{B}}

  \begin{observation}\label{obs:perlce}
    Assume that the factors $w[a,d_\AA-1]$ and $w[b,d_\BB-1]$ have the same
    string period, but neither $w[a,d_\AA]$ nor $w[b,d_\BB]$
    has this string period.
    Then $$\LCE(a,b)=\begin{cases}
    \min(d_\AA-a,d_\BB-b) & \text{if }d_\AA-a \ne d_\BB-b,\\
    d_\AA-a+\LCE(d_\AA,d_\BB) & \text{otherwise.}
    \end{cases}$$
  \end{observation}
  
        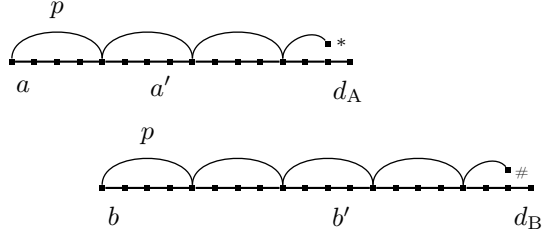
\begin{figure}[htpb]
    \begin{center}
    \begin{tikzpicture}[scale=1.2, auto,swap]
    \foreach \pos/\name in {{(1,0.2)/s}, {(2,0.2)/t}, {(3,0.2)/u},{(3.5,0.4)/a},
                           {(3,-1.2)/x}, {(4,-1.2)/y}, {(5,-1.2)/z}, {(5.5,-1.0)/d}}
      \node[vertex] (\name) at \pos {};
      \node[vertex] (r) at (0,0.2) {};
      \node[vertex] (v) at (1,-1.2) {};
      \node[vertex] (w) at (2,-1.2) {};
      \node[vertex] (b) at (3.5,0.2) {};
      \node[vertex] (c) at (5.5,-1.2) {};

 \node[label=below:$a$] at (0.125,0.2) {};
 \node[label=below:$b$] at (1.125,-1.2) {};
 \node[label=below:$d_{\mathrm{A}}$] at (3.725,0.2) {};
 \node[label=below:$d_{\mathrm{B}}$] at (5.725,-1.2) {};

   \foreach \x in {0.25,0.5,...,3.75} \draw (\x,0.2) node[vertex] {};
   \foreach \x in {1.25,1.5,...,5.75} \draw (\x,-1.2) node[vertex] {};
   \path (3.5,0.2) edge [arr] (3.75,0.2);
   \path (5.5,-1.2) edge [arr] (5.75,-1.2);
   \draw (3.65,0.2) node[above] {\scriptsize{*}};
   \draw (5.65,-1.2) node[above] {\tiny{\#}};
 
   \node[label=below:$a'$] at (1.65,0.25) {};
   \node[label=below:$b'$] at (3.65,-1.2) {};
      
 	\foreach \source/ \dest in {s/t, t/u, w/x, x/y, y/z}
        \path (\source) edge [arr, out=90, in=90, line width=0.5]   node{} (\dest);
    \foreach \source/ \dest in {r/s, s/t, t/u, u/b, v/w, w/x, x/y, y/z, z/c}
    	\path (\source) edge [arr]   node{} (\dest);
    \path (r) edge [arr, out=90, in=90, line width=0.5]   node[label=above:$p$]{} (s);
    \path (v) edge [arr, out=90, in=90, line width=0.5]   node[label=above:$p$]{} (w);
    \path (u) edge [arr, out=90, in=130, line width=0.5]   node{} (a);
    \path (z) edge [arr, out=90, in=130, line width=0.5]   node{} (d);

\end{tikzpicture}
    \end{center}
    \caption{
      In this example figure $d_\AA-a=14$, $d_\BB-b=18$, and $p=4$.
      We have $\LCE(a,b)=14$ and $\LCE(a',b')=8+\LCE(d_\AA,d_\BB)$.
    }\label{fig:fajny}
    \end{figure}

  \paragraph{\bf Non-Crossing Pairs}
  For a positive integer $n$, we define the set of pairs 
  $$P_n = \{(a,b)\in \mathbb{Z}^2 : 1\le a\le b \le n\}.$$
  Pairs $(a,b)$ and $(a',b')$ are called \emph{crossing} if $a < a' < b < b'$ or $a' < a < b' < b$.
  A subset $S \sub P_n$ is called \emph{non-crossing} if it does not contain crossing pairs.
   
   A graph $G$ is called \emph{outerplanar} if it can be drawn on a plane without crossings 
   in such a way that all vertices belong to the unbounded face.
   An outerplanar graph on $n$ vertices has less than $2n$ edges (at most $2n-3$ for $n\ge 2$).
  
    \begin{fact}\label{fct:ncb}
    A non-crossing set of pairs $S\subseteq P_n$ has less than $3n$ elements.
  \end{fact}
  \begin{proof}
  We associate $S\setminus \{(a,a): 1\le a \le n\}$ with a plane graph on vertices $\{1,\ldots,n\}$ drawn on a circle
  in this order, and edges represented as straight-line segments. 
  The non-crossing property of pairs implies that these segments do not intersect.
  Thus, the graph drawing is outerplanar, and therefore the number of edges is less than $2n$.
  Accounting for the pairs of the form $(a,a)$, we get the claimed upper bound.
  \qed\end{proof}

  For a set of pairs $S=\{(a_i,b_i): 1\le i \le k\}$ and a positive integer $t$, by $\ceil{S/t}$ we denote the set
  $\{(\ceil{\frac{a_i}{t}},\ceil{\frac{b_i}{t}}): 1\le i \le k\}$.

  \begin{observation}\label{obs:inherit}
    If $S$ is non-crossing, then $\ceil{S/t}$ is also non-crossing.
  \end{observation}  

\section{High-Level Description of the Data Structure}\label{sec:overview}
  We say that a sequence of $\LCE(a,b)$ queries, for $a \le b$, is \emph{non-crossing} if the underlying collection of pairs $(a,b)$ is non-crossing.
  In this section, we give an overview of our data structure, which answers a sequence of $q$ non-crossing $\LCE$ queries
  on-line in $\Oh(q+n\cdot \alpha(n))$ total time. 
  
  The data structure is composed of $\ceil{\log n}$ levels.
  Function $\LevelLCE{i}(a,b)$ corresponds to the level $i$
  and returns $\LCE(a,b)$.
  In the computation it may make calls to $\LevelLCE{i+1}(a,b)$.
  However, we make sure that the total number of such calls is bounded.
  Each original $\LCE(a,b)$ query is first asked at the level 0.

  The implementation of $\LevelLCE{i}(a,b)$ consists of two phases.
  If $\LCE(a,b) \ge 3\cdot 2^i$, then this $\LevelLCE{i}$ query is called \emph{relevant};
  otherwise it is called \emph{short}.
  In the first phase, we check the type of the query via a $\LimitedLCE_{\le 3\cdot 2^i}(a,b)$ query.
  This lets us immediately answer short queries.
  In the second phase, we know that the query is relevant, and we try to deduce the answer based on data gathered while processing \emph{similar} queries
  or to learn some information useful for answering future \emph{similar} queries by asking $\LevelLCE{i+1}$ queries.
  
  We shall say that $\LevelLCE{i}$ queries for $(a,b)$ and $(a',b')$ are \emph{similar} if $\lceil{\frac{a}{2^i}}\rceil=\lceil{\frac{a'}{2^i}}\rceil$
  and $\lceil{\frac{b}{2^i}}\rceil=\lceil{\frac{b'}{2^i}}\rceil$. 
  Each equivalence class of this relation is processed by an independent component, called a \emph{block-pair},
  identified by a pair of \emph{blocks} $(A,B)$,
  which are intervals of the form $[x\cdot 2^i+1, (x+1)\cdot 2^i]$ containing
  indices $a$ and $b$, respectively.
  If a relevant $\LevelLCE{i}(a,b)$ query satisfies $a\in A$ and $b\in B$ for some block-pair $(A,B)$, 
  we say that the block-pair is \emph{responsible} for the query
  or that the query \emph{concerns} the block-pair.
  As we show in \cref{sec:comp}, the pairs of interval right endpoints of block-pairs at each level are non-crossing
  (whereas $\LevelLCE{i}$ queries that will be asked for $i \ge 1$ are non necessarily non-crossing).

  The implementation of a block-pair, summarized in the lemma below, is given in \cref{sec:bp}.
  
  \begin{restatable}{lemma}{lembp}\label{lem:bp}
  Consider a sequence of relevant $\LevelLCE{i}$ queries concerning a block-pair $(A,B)$.
  The block-pair can answer these queries on-line in worst-case constant time plus
  the time to answer at most four $\LevelLCE{i+1}(a,b)$ queries,
  such that each either corresponds to the currently processed $\LevelLCE{i}$ query or satisfies $a< b \le a+2^{i+1}$.
  \end{restatable}
   
  Structural conditions stated in \cref{lem:bp} let us characterize the set of queries passed to the next level. 
  The complexity analysis in \cref{sec:comp} relies on this characterization.

    \newcommand{\A}{\mathbf{initial}}
    \newcommand{\B}{\mathbf{visited}}
    \newcommand{\C}{\mathbf{full}}
    \newcommand{\D}{\mathbf{full}^+}
    \newcommand{\state}{\textrm{state}}

    \section{Block-Pair Implementation}\label{sec:bp}
    Our aim in this section is to prove \cref{lem:bp}.
    Information stored by a block-pair changes through the course of the algorithm,
    and the implementation of the query algorithm depends on what is currently stored.
    We distinguish four states of a block pair $(A,B)$ at level $i$.
    \cref{fig:state-b-c} illustrates two of the states.

    \begin{center}
    \renewcommand*{\arraystretch}{1.5}
      \begin{tabular}{c|>{\raggedright\arraybackslash}p{9cm}<{}}
        $\state(A,B)$ & \emph{description}\\ \hline
            $\A$ & No auxiliary data is stored.\\
    $\B(a_0,b_0,L)$ & $a_0\in A$, $b_0\in B$, $L=\LCE(a_0,b_0)\ge 3\cdot 2^i$.\\
     $\C(d_\AA,d_\BB)$ & $\exists_{p\in [1,2^{i+1}]} \; : $
     $w[\max A , d_\AA-1]$ and $w[\max B, d_\BB -1]$ have common period $p$ and length at least $p+2^i$,
     but neither $w[\max A, d_\AA]$ nor  $w[\max B , d_\BB ]$ has period $p$.\\
 $\D(d_\AA, d_\BB,L')$ & As in $\C(d_\AA,d_\BB)$ plus $L'=\LCE(d_\AA,d_\BB)$.
      \end{tabular}
    \end{center}


    \begin{figure}
    \begin{center}
    \begin{tikzpicture}[scale=0.6]
\tikzstyle{v} = [draw, circle, fill=black, minimum size = 3pt, inner sep = 0 pt, color=black]

\node at (-1,1.8) {$\B(a_0,b_0,L)$};

\draw[densely dashed] (-1,0) -- (-0.25,0)  (8.25,0) -- (9.25,0)  (17.25,0) -- (18.25,0);
\draw[yshift=0.4cm,densely dashed] (-1,0) -- (-0.25,0)  (8.25,0) -- (9.25,0)  (17.25,0) -- (18.25,0);

\begin{scope}
\draw[xshift=-1cm,fill=black!10] (1,0) -- (2.5,0) -- (2.5,0.4) -- node[above]{$A$} (1,0.4) -- cycle;
\node (ap) [v] at (0.5,0) {}; \node at (0.5,0) [below] {$a_0$};
\draw (-0.25, 0) -- (8.25,0) (-0.25, .4) -- (8.25,.4);
\draw [|<->|] (0.5,-1)--node[midway,fill=white] {$L=\LCE(a_0,b_0)$} +(6.5,0);
\draw (7.3,-1) node {\footnotesize$*$};
\end{scope}

\begin{scope}[xshift = 10cm]
\draw[xshift=-1.5cm,fill=black!10] (1,0) -- (2.5,0) -- (2.5,0.4) -- node[above]{$B$} (1,0.4) -- cycle;
\draw (-0.75, 0) -- (7.5,0) (-0.75, .4) -- (7.5,.4);
\node (bp) [v] at (0.5,0) {}; \node at (0.5,0) [below] {$b_0$};
\draw [|<->|] (0.5,-1)--node[midway,fill=white] {$L=\LCE(a_0,b_0)$} +(6.5,0);
\draw (7.3,-1) node {\scriptsize\#};
\end{scope}

\end{tikzpicture}
    \bigskip
    \begin{tikzpicture}[scale=0.6]
\tikzstyle{v} = [draw, circle, fill=black, minimum size = 3pt, inner sep = 0 pt, color=black]

\node at (-1,1.8) {$\D(d_\AA,d_\BB,L')$};

\draw[densely dashed] (-1,0) -- (-0.25,0)  (8.25,0) -- (9.25,0)  (17.25,0) -- (18.25,0);
\draw[yshift=0.4cm,densely dashed] (-1,0) -- (-0.25,0)  (8.25,0) -- (9.25,0)  (17.25,0) -- (18.25,0);

\begin{scope}
\draw[xshift=-1cm,fill=black!10] (1,0) -- (2.5,0) -- (2.5,0.4) -- node[above]{$A$} (1,0.4) -- cycle;
\node (ap) [v] at (6.4,0) {}; \node at (6.4,0) [below] {$d_\AA$};
\draw (-0.25, 0) -- (8.25,0) (-0.25, .4) -- (8.25,.4);
\begin{scope}[xshift=1.5cm, yshift=-0.6cm]
\clip(0, 1) rectangle (4.7, -1);
\foreach \x in {0,1,2,3,4} {
        \draw (\x*1, 0) sin (\x*1+.5, -0.2) cos (\x*1 + 1, 0);
}
\end{scope}

\draw [|<->|] (1.5,-1.6)--node[midway,fill=white,anchor=mid] {$p$} +(1,0);
\draw [|<->|] (2.5,-1.6)--node[midway,fill=white,anchor=mid] {$\ge 2^i$} (6.3,-1.6);
\draw [|<->|] (6.3,-1.1) --node[midway,fill=white,anchor=mid] {$L'$} +(1.5cm,0);
\draw (8.1,-1.1) node {\footnotesize$*$};

\end{scope}

\begin{scope}[xshift = 10cm]
\draw[xshift=-1.5cm,fill=black!10] (1,0) -- (2.5,0) -- (2.5,0.4) -- node[above]{$B$} (1,0.4) -- cycle;
\draw (-0.75, 0) -- (7.5,0) (-0.75, .4) -- (7.5,.4);
\node (bp) [v] at (5.7,0) {}; \node at (5.7,0) [below] {$d_\BB$};

\begin{scope}[xshift=-.7cm,yshift=-0.6cm]
\clip(1.7, 0) rectangle (6.4, -1);
\foreach \x in {0,1,2,3,4,5,6,7} {
        \draw (\x*1, 0) sin (\x*1+.5, -0.2) cos (\x*1 + 1, 0);
}
\node at (0.5,0.2) [above] {$p$};
\end{scope}

\draw [|<->|] (1,-1.6)--node[midway,fill=white,anchor=mid] {$p$} +(1,0);
\draw [|<->|] (2,-1.6)--node[midway,fill=white,anchor=mid] {$\ge 2^i$} (5.7,-1.6);
\draw [|<->|] (5.7,-1.1) --node[midway,fill=white,anchor=mid] {$L'$} +(1.5cm,0);

\draw (7.5,-1.1) node {\scriptsize\#};

\end{scope}

\end{tikzpicture}
    \end{center}
    \caption{Block-pair $(A,B)$ in states $\B(a_0,b_0,L)$ and $\D(d_\AA,d_\BB,L')$.
    }\label{fig:state-b-c}
    \end{figure}
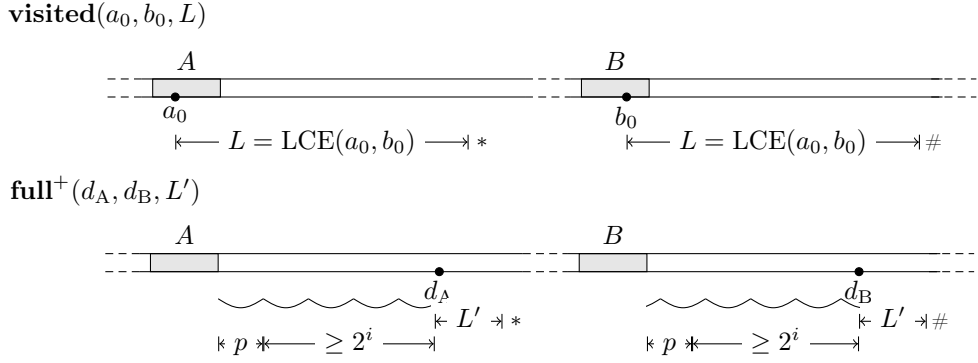

  \subsection{Initial State}

  In this state, we simply forward the query to the level $i+1$, return the obtained $\LCE(a,b)$ value,
  and change the state to $\B(a,b,\LCE(a,b))$.

    \begin{algorithm}[H]
      \caption{$\mathrm{Initial\mbox{-}LCE}^{(i)}_{(A,B)}(a,b)$}\label{alg:a}
    \KwRequire{$\LevelLCE{i}(a,b)$ concerns $(A,B)$, whose state is $\A$}
    $L \assign \LevelLCE{i+1}(a,b)$\Comment*{higher level call}
    transform $(A,B)$ to state $\B(a,b,L)$\;
    \KwRet{$L$}\;
    \end{algorithm}

  \subsection{Visited State}

  In state $\B(a_0,b_0,L)$, we can immediately determine $\LCE(a,b)$
  if $(a,b)$ is a shift of $(a_0,b_0)$.
  Otherwise, we apply \cref{lem:btoc} to move to state $\C$.

  \begin{lemma}\label{lem:btoc}
  Let $\LevelLCE{i}(a,b)$, $\LevelLCE{i}(a',b')$ be similar and relevant queries
   and let $p=|(b-b')-(a-a')|$.
  If $p\ne 0$ and $b'\le b$,  then $\LCE(a,a+p)\ge 2^{i+1}$, i.e., $p$~is a (not necessarily shortest) period of the factor $w[a,a+2^{i+1}+p-1]$.
  \end{lemma}
  \begin{proof}
  We shall first prove that $\LCE(a,a+q)\ge 3\cdot 2^i-(b-b')$ where $q=(b-b')-(a-a')$. 
  First, observe that $a+q=a'+(b-b')$, and thus
  $\LCE(a+q,b)=\LCE(a'+(b-b'),b'+(b-b'))\ge 3\cdot 2^i-(b-b')$ because $\LevelLCE{i}(a',b')$ is relevant.
  Since $\LevelLCE{i}(a,b)$ is also relevant,
  we have $\LCE(a,b)\ge 3\cdot 2^i \ge 3\cdot 2^i-(b-b')$.
  Combining these two inequalities, we immediately get $\LCE(a,a+q)\ge \min(\LCE(a,b),\LCE(a+q,b)) \ge 3\cdot 2^i-(b-b')$,
  as claimed.
  
  If $q > 0$, we have $q=p$, and thus $\LCE(a,a+p)\ge 3\cdot 2^i - (b-b')$.
  Since the two $\LevelLCE{i}$ queries are similar, we have $ 3\cdot 2^i - (b-b') \ge 2^{i+1}$,
  so $\LCE(a,a+p)\ge 2^{i+1}$.
  See \cref{fig:btoc} for an illustration of this case.
  
  Otherwise, $q=-p$, and we have $\LCE(a,a-p)\ge 3\cdot 2^i - (b-b')$,
  which implies $\LCE(a+p,a)\ge 3\cdot 2^i-(b-b')+q = 3\cdot 2^i-(a-a')$.
  Again, the fact that the queries are similar yields
  $3\cdot 2^i-(a-a') \ge 2^{i+1}$, and consequently $\LCE(a,a+p)\ge 2^{i+1}$.
  \qed\end{proof}

\begin{figure}[ht]
  \centering
  \vspace{1cm}
  \begin{tikzpicture}

\tikzstyle{v} = [draw, circle, fill=black, minimum size = 3pt, inner sep = 0 pt, color=black]

\newcommand\LceLine[3]{
  \draw (#1,#2)--+(0,0.2)--+(#3,0.2)--+(#3,0)--cycle;
}
\newcommand\Mark[3]{
  \draw [pattern=north west lines, pattern color=black] (#1,#2)--+(0,0.2)--+(#3,0.2)--+(#3,0)--cycle;
}

\draw[densely dashed] (-0.7, 0) -- (11.9,0.0);
\draw[densely dashed] (-0.7, 0.2) -- (11.9,0.2);

\begin{scope}
\draw (-.35,0)--(5.65,0);
\draw (-.35,.2)--(5.65,.2);
\draw[fill=black!10] (-.1, 0) rectangle (1.3, 0.2);

\node (ap) [v] at (0,0) {}; \node at (1.2,-.3) [anchor=mid] {$a'$};
\node (a) [v] at (1.2,0) {}; \node at (0,-.3) [anchor=mid] {$a$};
 \LceLine{0}{0.7}{4.2}
 \Mark{0}{0.7}{3.7}
 
\LceLine{1.2}{-0.9}{4.2}
\Mark{1.7}{-0.9}{3.7}

\begin{scope}[yshift=1.1cm]
\clip(0, 0) rectangle (5.4, 1);
\foreach \x in {0,...,9} {
        \draw (\x*1.7, 0) sin (\x*1.7+.85, 0.2) cos (\x*1.7 + 1.7, 0);
}
\node at (0.85,0.2) [above] {$p=q$};
\draw [|<->|,transform canvas={yshift=0.6cm}] (1.7,0)--node[midway,fill=white] {$\ge 2^{i+1}$} (5.4,0);

\end{scope}

\draw [|<->|,transform canvas={yshift=-1.4cm}] (1.2,0)--node[midway,yshift=-0.3cm,fill=white] {\small $b-b'$} +(.5,0);

\end{scope}

\begin{scope}[xshift=6.6cm]
\draw (-.55,0)--(4.95,0);
\draw (-.55,.2)--(4.95,.2);
\draw[fill=black!10] (-.3, 0) rectangle (1, 0.2);

\node (bp) [v] at (0,0) {}; \node at (0,-.3) [anchor=mid] {$b'$};
\node (b) [v] at (.5,0) {}; \node at (.5,-.3) [anchor=mid] {$b$};
 \LceLine{.5}{0.7}{4.2}
 \LceLine{0}{-0.9}{4.2}
 
 \Mark{.5}{-0.9}{3.7}
\Mark{.5}{0.7}{3.7}

\draw [|<->|,transform canvas={yshift=-1.4cm}] (bp.center)--node[midway,fill=white] {$\LCE(a',b')\ge 3\cdot 2^i$} +(4.2,0);
\draw [|<->|,transform canvas={yshift=1.4cm}] (b.center)--node[midway,fill=white] {$\LCE(a,b)\ge 3\cdot 2^i$} +(4.2,0);

\end{scope}

\end{tikzpicture}
  \vspace{1cm}

  \begin{tikzpicture}

\tikzstyle{v} = [draw, circle, fill=black, minimum size = 3pt, inner sep = 0 pt, color=black]

\newcommand\LceLine[3]{
  \draw (#1,#2)--+(0,0.2)--+(#3,0.2)--+(#3,0)--cycle;
}
\newcommand\Mark[3]{
  \draw [pattern=north west lines, pattern color=black] (#1,#2)--+(0,0.2)--+(#3,0.2)--+(#3,0)--cycle;
}
\draw[densely dashed] (-0.7, 0) -- (11.9,0.0);
\draw[densely dashed] (-0.7, 0.2) -- (11.9,0.2);

\begin{scope}
\draw (-.35,0)--(5.65,0);
\draw (-.35,.2)--(5.65,.2);
\draw[fill=black!10] (-.1, 0) rectangle (1.3, 0.2);

\node (ap) [v] at (0,0) {}; \node at (0,-.3) [anchor=mid] {$a'$};
\node (a) [v] at (1.2,0) {}; \node at (1.2,-.3) [anchor=mid] {$a$};
 \LceLine{0}{-0.9}{4.2}
\Mark{.5}{-0.9}{3.7}

 \LceLine{1.2}{0.7}{4.2}
\Mark{1.2}{0.7}{3.7}

\begin{scope}[xshift=1.2cm, yshift=1.1cm]
\clip(-1, 0) rectangle (3.7, 1);
\foreach \x in {-1,...,9} {
        \draw (\x*.7, 0) sin (\x*.7+.35, 0.2) cos (\x*.7 + .7, 0);
}
\node at (0.35,0.2) [above] {$p$};
\draw [|<->|,transform canvas={yshift=0.6cm}] (.7,0)--node[midway,fill=white] {$\ge 2^{i+1}$} (3.7,0);

\end{scope}

\draw [|<->|,transform canvas={yshift=-1.4cm}] (ap.center)--node[midway,yshift=-0.3cm,fill=white] {\small $b-b'$} +(.5,0);

\end{scope}

\begin{scope}[xshift=6.6cm]
\draw (-.55,0)--(4.95,0);
\draw (-.55,.2)--(4.95,.2);
\draw[fill=black!10] (-.3, 0) rectangle (1, 0.2);

\node (bp) [v] at (0,0) {}; \node at (0,-.3) [anchor=mid] {$b'$};
\node (b) [v] at (.5,0) {}; \node at (.5,-.3) [anchor=mid] {$b$};
 \LceLine{.5}{0.7}{4.2}
 \LceLine{0}{-0.9}{4.2}
 
 \Mark{.5}{-0.9}{3.7}
\Mark{.5}{0.7}{3.7}

\draw [|<->|,transform canvas={yshift=-1.4cm}] (bp.center)--node[midway,fill=white] {$\LCE(a',b')\ge 3\cdot 2^i$} +(4.2,0);
\draw [|<->|,transform canvas={yshift=1.4cm}] (b.center)--node[midway,fill=white] {$\LCE(a,b)\ge 3\cdot 2^i$} +(4.2,0);

\end{scope}

\end{tikzpicture}
  \vspace{1cm}

  \caption{Illustration of \cref{lem:btoc}: case $q>0$.
    Illustration of \cref{lem:btoc}:
    upper part corresponds to case $q>0$, lower part to case $q<0$.
    In both cases we assume that $\LCE(a+q,b) \le \LCE(a,b)$.
    The marked fragments correspond to $\LCE(a,a+q)=\LCE(a+q,b)$.
  }\label{fig:btoc}
  \end{figure}
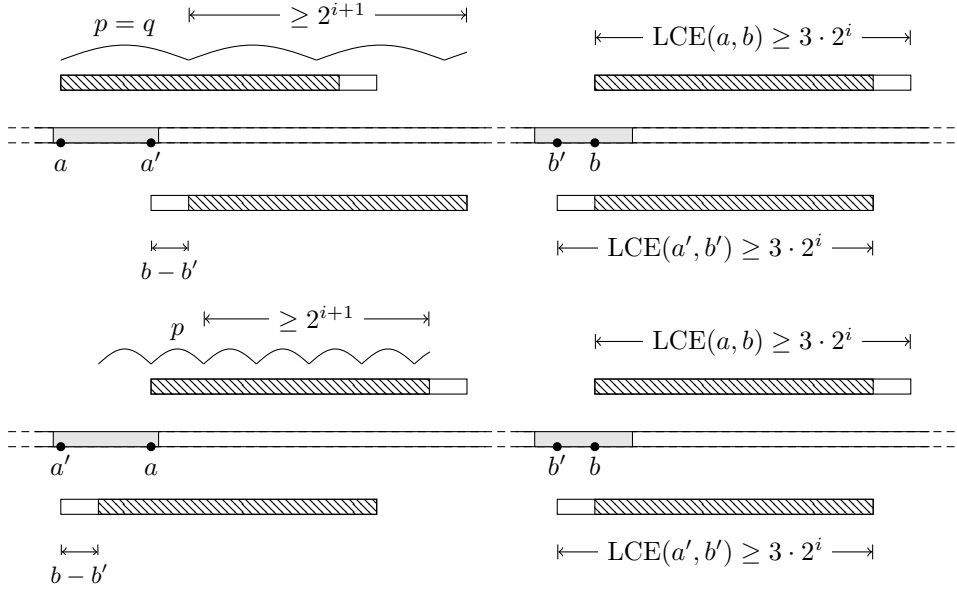
    
  In the query algorithm, we first check if $a-a_0=b-b_0$.
  If so, let us denote the common value by $\Delta$.
  Note that $|\Delta|\le 2^i$, $\LCE(a,b)\ge 3\cdot 2^i$, and $\LCE(a_0,b_0)\ge 3\cdot 2^i$.
  This clearly yields $\LCE(a,b)=\LCE(a_0,b_0)+\Delta$, which lets us compute the result in constant time.

  \medskip
     \begin{algorithm}[H]
       \caption{$\mathrm{Visited\mbox{-}LCE}^{(i)}_{(A,B)}(a,b)$}\label{alg:b}
     \KwRequire{$\LevelLCE{i}(a,b)$ concerns $(A,B)$, whose state is $\B(a_0,b_0,L)$
    }
    \eIf{$a-a_0=b-b_0$}{
      \KwRet{$L+a-a_0$}\;
    }{
      $p \assign |(a-a_0)-(b-b_0)|$\;
      $a ' \assign \max A$;\ \ $b ' \assign \max B$\;
      $d_\AA \assign a'+p+ \LevelLCE{i+1}(a', a'+p)$\Comment*{higher level call}
      $d_\BB \assign b'+p + \LevelLCE{i+1}(b', b'+p)$\Comment*{higher level call}
      transform $(A,B)$ to state $\C(d_\AA,d_\BB)$\;
      \KwRet{$\mathrm{Full\mbox{-}LCE}^{(i)}_{(A,B)}(a,b)$}\Comment*{recursive call on state $\C$}
    }
    \end{algorithm}

    \smallskip
  Otherwise, our aim is to change the state of the block-pair to $\C$.
  \cref{lem:btoc} lets us deduce that $\LCE(\bar{a},\bar{a}+p)\ge 2^{i+1}$ for some $\bar{a}\in \{a,a_0\}$
  and (by symmetry) $\LCE(\bar{b},\bar{b}+p)\ge 2^{i+1}$ for some $\bar{b}\in \{b,b_0\}$,
  where $p=|(a-a_0)-(b-b_0)|$
  ($\bar{a}$ and $\bar{b}$ depend on the relative order of $b,b_0$ and $a,a_0$, respectively).
  Let $a'=\max A$ and $b'=\max B$.
  We have $\LCE(a',a'+p)\ge 2^i$ and $\LCE(b',b'+p)\ge 2^i$ because $a'-2^i < \bar{a}\le a'$
  and $b'-2^i < \bar{b}\le b'$.
  Such a situation allows for a move to state $\C$. 
  The exact values of $d_\AA$ and $d_\BB$ are computed using a higher level call, which lets us determine $\LCE(a',a'+p)$ and $\LCE(b',b'+p)$.
  Note that $p\le 2^{i+1}$ implies that these queries satisfy the condition of \cref{lem:bp}.
  The answer to the initial $\LevelLCE{i}(a,b)$ query is computed by the routine for state $\C$, which we give below.

  \subsection{Full States}
  In state $\D$ we can answer every relevant query in constant time.
  In state $\C$ we can either answer the query in constant time or make the final query at level $i+1$
  to transform the state to $\D$; see the following lemma.
  
  \begin{lemma}\label{lem:c}
  Consider a relevant $\LevelLCE{i}(a,b)$ query concerning a block-pair $(A,B)$
  in state $\C(d_\AA,d_\BB)$ or $\D(d_\AA,d_\BB,L')$.
    Then     $$\LCE(a,b)=\begin{cases}
  \min(d_\AA-a,d_\BB-b) & \text{if }d_\AA-a\ne d_\BB-b,\\
  d_\AA-a + \LCE(d_\AA,d_\BB) & \text{otherwise}.
  \end{cases}$$
  \end{lemma}
  \begin{proof} 
  Let $a_0=\max A$, $b_0=\max B$ and let $p$ be the witness period of the state of $(A,B)$.
  Let us define $\Delta = \max(a_0-a,b_0-b)$, $a'=a+\Delta$, and $b'=b+\Delta$.
  Observe that $\Delta\le 2^i$, $a_0\le a'\le a_0+2^{i}$, and $b_0\le b' \le b_0+2^{i}$.
  The fact that the query is relevant yields $\LCE(a,b)\ge 3\cdot 2^i \ge p + \Delta$, so $\LCE(a,b)=\Delta + \LCE(a',b')$
  and $\LCE(a',b')\ge p$. 
    Moreover, $d_\AA\ge p+2^{i}+a_0$ and $d_\BB\ge p+2^{i}+b_0$ implies that fragments $w[a',d_{\AA}-1]$ and $w[b',d_{\BB}-1]$
  have length at least $p$, and thus they are right-maximal with period $p$.
  Consequently, the fragments $w[a',d_{\AA}-1]$ and $w[b',d_{\BB}-1]$ have the same string period of length~$p$.
  This lets us apply \Cref{obs:perlce}, which gives
   $$\LCE(a',b')=\begin{cases}
  \min(d_\AA-a',d_\BB-b') & \text{if }d_\AA-a'\ne d_\BB-b',\\
  d_\AA-a' + \LCE(d_\AA,d_\BB) & \text{otherwise}.
  \end{cases}$$
  Since $a'=a+\Delta$, $b'=b+\Delta$, and  $\LCE(a,b)=\Delta + \LCE(a',b')$,
  this is clearly equivalent to the claimed formula for $\LCE(a,b)$.
  \qed\end{proof}

    \begin{algorithm}[H]
      \caption{$\mathrm{Full\mbox{-}LCE}^{(i)}_{(A,B)}(a,b)$}\label{alg:c}
    \KwRequire{$\LevelLCE{i}(a,b)$ concerns $(A,B)$, whose state is $\C(d_\AA,d_\BB)$
    or $\D(d_\AA,d_\BB,L')$}

    \vspace*{0.3cm}
    \If{$d_\AA-a \ne d_\BB-b$}{\KwRet{$\min(d_\AA-a,d_\BB-b)$}\;}
    \Else{
        \If{$(A,B)$ is in state $\C(d_\AA,d_\BB)$}{
          $L' \assign \LevelLCE{i+1}(a,b)-(d_\AA-a)$\Comment*{higher level call}
         transform $(A,B)$ to state $\D(d_\AA, d_\BB,L')$\;
        }
    	\KwRet{$d_\AA-a+L'$}\;
    }
    \end{algorithm}

  \subsection{Proof of \cref{lem:bp}}
  
  \lembp*
  \begin{proof}
  \cref{alg:a,alg:b,alg:c} answer queries concerning the block-pair $(A,B)$, and use constant time.
  The level $i+1$ call is only made when the state changes. 
  The original query is forwarded during a shift from state $\A$ to $\B$ and from state $\C$ to $\D$, 
  while during a shift from $\B$ to $\C$ two LCE queries are asked, both with arguments at distance $p\le 2^{i+1}$, as claimed.
  \qed\end{proof}

    \section{Complexity Analysis}\label{sec:comp}
    \cref{alg:overview} summarizes the implementation of the $\LevelLCE{i}(a,b)$ function.
    As mentioned in \cref{sec:overview}, we first compute $\LimitedLCE_{\le 3\cdot 2^i}(a,b)$,
    which might immediately give us the sought value $\LCE(a,b)$.
    Otherwise the query is relevant, and we refer to the block-pair $(A,B)$ which is responsible for the query.
   
        \begin{algorithm}[h]
          \caption{$\LevelLCE{i}(a,b)$}\label{alg:overview}
    $\ell \assign \LimitedLCE_{\le 3\cdot 2^i}(a,b)$\;
    \eIf(\Comment*[f]{short query}){$\ell < 3\cdot 2^i$}{%
      \KwRet{$\ell$}\;
    }(\Comment*[f]{relevant query}){%
      $(A,B) \assign $ block-pair responsible for the query $(a,b)$ at level $i$\;
      \KwSty{return}:\\
      \begin{tabular}{p{0.3cm}p{4cm}p{5cm}}
      & $\mathrm{Initial\mbox{-}LCE}^{(i)}_{(A,B)}(a,b)$ & if $(A,B)$ is in state $\A$ \\
      & $\mathrm{Visited\mbox{-}LCE}^{(i)}_{(A,B)}(a,b)$ & if $(A,B)$ is in state $\B$ \\
      & $\mathrm{Full\mbox{-}LCE}^{(i)}_{(A,B)}(a,b)$ & if $(A,B)$ is in state $\C$ or $\D$ \\
      \end{tabular}
    }
    \end{algorithm}

  \noindent
  Let $S_i = \{ (a,b) : \LevelLCE{i}(a,b)\text{ is called}\,\}$.
  Then $\ceil{S_i / 2^{i}}$ corresponds to the set of pairs of interval right endpoints of block-pairs at level $i$.

  \begin{fact}\label{fct:nci}
    The set $\ceil{S_i / 2^{i}}$ is non-crossing.
  \end{fact}
  \begin{proof}
    We proceed by induction on $i$. The base case is trivial from the assumption on the input sequence.
    \cref{lem:bp} proves that $S_{i+1}\sub S_i \cup \{(a,b) : a < b \le a+2^{i+1}\}$.
    Hence, 
    $\ceil{S_{i+1}/2^{i+1}}\sub \ceil{\ceil{S_{i}/2^i}/2} \cup \{(a,b): a\le b \le a+1\}$.
    The first component is non-crossing by the inductive hypothesis combined with \cref{obs:inherit}.
    Pairs of the form $(a,a)$ and $(a,a+1)$ do not cross any other pair, so adding them to a non-crossing
    family preserves this property. 
  \qed\end{proof}
  
  \noindent
  Consequently, \cref{fct:ncb} proves that the number of block-pairs responsible for a query at level $i-1$
  is bounded by $\frac{3n}{2^{i-1}}$.
  Each of them yields at most 4 queries at level $i$.
  This leads straight to the following bound.
  
  \begin{observation}\label{obs:few}
  $|S_{i}|\le \frac{24n}{2^{i}}$ for $i\ge 1$.
  \end{observation}

    \noindent
  If we stored the block-pairs using a hash table, we could retrieve the internal data of the block-pair responsible for $(a,b)$ 
  in randomised constant time.
  However, in the case of non-crossing $\LCE$ queries we can make this time worst-case.
  
  Recall from \cref{fct:ncb} that for a set $S \sub P_n$ of non-crossing pairs
  we can identify $S\setminus \{(a,a): 1\le a \le n\}$ with an outerplanar graph on vertices $\{1,\ldots,n\}$.
  We say that a simple undirected graph has \emph{arboricity} at most $c$ if it can be partitioned into $c$ forests.
  Outerplanar graphs have arboricity at most 2 (see \cite{NashWilliams}) which lets us use the following theorem
  to store $S\setminus \{(a,a): 1\le a \le n\}$.
  Membership queries for pairs $(a,a)$ are trivial to support using an array.

\begin{theorem}[\cite{DBLP:conf/wads/BrodalF99}]\label{thm:arb}
Consider a graph of arboricity $c$ with vertices given in advance and edges revealed on-line.
One can support adjacency queries, asking to return the edge between two given vertices or \textbf{nil} if it does not exist, in worst-case $\Oh(c)$ time,
with edge insertions processed in amortized constant time.
\end{theorem}
  
  \noindent
  The following corollary shows, by \cref{fct:nci}, that indeed the block-pairs at each level can be retrieved
  in worst-case constant time.

\begin{corollary}\label{cor:get}
Consider a set $S\sub P_n$ of non-crossing pairs arriving on-line.
One can support membership queries (asking if $(a,b)\in S$ and, if so, to return data associated with this pair)
in worst-case constant time with insertions processed in amortized constant time.
\end{corollary}

    \begin{theorem}\label{thm:main}
    In a string of length $n$, a sequence of $q$ non-crossing LCE queries can be answered in total time $\Oh(q+n\cdot \alpha(n))$.
    \end{theorem}
    \begin{proof}
    For $i>0$, an $\LevelLCE{i}$ query, excluding the $\LevelLCE{i+1}$ queries called, requires $\Oh(i\cdot \alpha(n))$ time
    for answering a $\LimitedLCE_{\le 3 \cdot 2^i}$ query by \cref{lem:CPM}
    plus $O(1)$ additional time by \cref{lem:bp}. For $i=0$ we may compute $\LimitedLCE_{\le 3}$ na\"{i}vely in constant time,
    so the running time is constant.
    
    The number of $\LevelLCE{0}$ queries is $q$, while the number of 
    $\LevelLCE{i}$ queries for $i\ge 1$ is $\Oh(\frac{n}{2^i})$ by \cref{obs:few}.    
    The total running time is therefore
    $$
      \Oh\bigg(q+n\cdot \alpha(n) \cdot \sum_{i=1}^{\infty} \frac{i}{2^{i}}\bigg) = \Oh(q+n\cdot \alpha(n)).
    $$
    
    \vspace*{-1cm}\qed\end{proof}

  \section{Computing Runs}\label{sec:runs}
  Bannai et al. \cite{DBLP:journals/corr/BannaiIINTT14,DBLP:conf/soda/BannaiIINTT15} presented an algorithm
  for computing all the runs in a string of length $n$ that works in time proportional to answering $\Oh(n)$ LCE queries
  on the string or on its reverse.
  As main tool they used Lyndon trees.
  We note here that the LCE queries asked by their algorithm can be divided into a constant number of groups,
  each consisting of non-crossing LCE queries.
  Roughly speaking, this is based on the obvious fact that intervals in a Lyndon tree form a \emph{laminar family}, i.e.,
  for every two they are either disjoint or one of them contains the other.

  In the first phase, given a string $w$, the algorithm of \cite{DBLP:journals/corr/BannaiIINTT14,DBLP:conf/soda/BannaiIINTT15}
  constructs $\LTree_0(w)$ and $\LTree_1(w)$.
  For each $\idx \in \{0,1\}$, the construction of $\LTree_\idx(w)$ goes from right to left.
  Before the $k$-th step (for $k=n,\ldots,1$), we store on a stack the roots of subtrees of $\LTree_\idx(w)$
  that correspond to $w[k+1,n]$.
  Hence, the intervals corresponding to the roots on the stack are disjoint and cover the interval $[k+1,n]$.
  In the $k$-th step we push on the stack a single node corresponding to $[k,k]$.
  Afterwards, as long as the stack contains at least two elements and the top element $[k,l]$ and the second to top element $[a,b]$
  satisfy $w[k,l] \prec_\idx w[a,b]$, we pop the two subtrees from the stack and push one subtree with the root $[k,b]$.
  The lexicographical comparison is performed via an $\LCE(k,a)$ query.
  
  \begin{observation}\label{obs:ph1}
    The $\LCE$ queries asked in the construction of $\LTree_\idx(w)$ are non-crossing.
  \end{observation}
  \begin{proof}
    In the $k$-th step of the algorithm we only ask $\LCE(i,j)$ queries for $i=k$.
    Suppose towards contradiction that in the course of the algorithm we ask two LCE queries with $(i,j)$ and $(i',j')$
    such that $i<i'<j<j'$.
    The latter is asked at step $i'$, and at that moment $[i',j'-1]$ is a root of a subtree of $\LTree_\idx(w)$.
    Then the former is asked at step $i$, and then $[i,j-1]$ is a root of a subtree of $\LTree_\idx(w)$.
    This contradicts the fact that the intervals in $\LTree_\idx(w)$ form a laminar family.
  \qed\end{proof}

  In the second phase, for each node $[a,b]$ of each Lyndon tree $\LTree_\idx(w)$ we
  check if there is a run with period $p=b-a+1$ that contains $w[a,b]$.
  To this end we check how long does the periodicity with period $p$ extend to the right and to the left of $w[a,b]$.
  The former obviously reduces to an $\LCE(a,b+1)$ query and the latter to an LCE query in the reverse of $w$,
  which is totally symmetric.
  As the intervals in $\LTree_\idx(w)$ form a laminar family, we arrive at the following.

  \begin{observation}\label{obs:ph2}
    The $\LCE$ queries asked when right-extending the periodicity of the intervals from $\LTree_\idx(w)$ are non-crossing.
  \end{observation}

  By \cref{obs:ph1,obs:ph2}, \cref{thm:main} yields the following result and its immediate corollary.

  \begin{theorem}
    The Lyndon tree and the set of all runs in a string of length $n$ over a general ordered alphabet can be computed in $\Oh(n \alpha(n))$ time.
  \end{theorem}

  \begin{corollary}
    All the local periods and the number of all squares in a string of length $n$ over a general ordered alphabet can be computed in $\Oh(n \alpha(n))$ time.
  \end{corollary}

  \bibliographystyle{splncs03}
  \bibliography{arxiv}

\newpage
\appendix

  \section*{Appendix}

\subsection*{Example of a Run}
The figure below shows an example of a run with period 3\ in a string.
This string contains also other runs, e.g.\ $w[10,12]$ with period 1 and $w[1,5]$ with period 2.

  \begin{figure}
    \begin{center}
      \begin{tikzpicture}

\foreach [count=\x]  \ch in {a,b, a,b,a, a,b,a, a, b,b,b,a,a} {
  \node at (\x*0.5,0) [above] {\tt \ch};
}
\node at (0,0.25) [left] {$w=$};
\node at (1.5,0) [below] {\tiny 3};
\node at (5,0) [below] {\tiny 10};

\begin{scope}[xshift=1.3cm,yshift=0.5cm]
\clip(0, 0) rectangle (4, 1);
\foreach \x in {0,1,2,3,4} {
  \draw (\x*1.5, 0) sin (\x*1.5+.75, 0.3) cos (\x*1.5 + 1.5, 0);  
}
\end{scope}

\end{tikzpicture}
    \end{center}
    \caption{Example of a run $w[3,10]$ in the string $w=\mathtt{ababaabaabbbaa}$.}\label{fig:example-run}
  \end{figure}

\subsection*{Example of a Lyndon Tree}
A Lyndon tree of a Lyndon word is obtained by applying standard factorization recursively on the Lyndon word.
The figure below presents an example of a Lyndon tree.
In the algorithm each node stores an interval describing the factor of $w$ that it corresponds to.

  \begin{figure}
    \begin{center}
          \scalebox{0.7}{
\begin{forest}
  for tree={
    internal,
    if n children=0{tier=terminal, 
    before typesetting nodes={
      label/.wrap pgfmath arg={below:{\Large\leaftext{#1}}}{content()},
          content={},
          leaf,
        },
    }
    {},
    edge path={
        \noexpand\path[\forestoption{edge}]
        (!u.parent anchor) -- ([xshift=0mm].child anchor);},
    s sep=30pt,
  },
  [
    [a]
    [, tier=t1
	  [, level=3
        [, tier=t5
          [a]
          [, tier=t7
            [a]
            [b]
          ]
        ]        
        [,tier=t6
          [a]
          [b]
        ]            
      ]
      [, level=2
        [, tier=t4
          [a]
          [, tier=t5
            [, tier=t6
              [a]
              [b]
            ]
            [b]
          ]
        ]
        [, tier=t5
          [, tier=t6
            [a]
            [b]
          ]
          [b]
        ]     
      ]
    ]
  ]
\end{forest}
}
    \end{center}
    \caption{The Lyndon tree $\LTree_0(w)$ of a Lyndon word $w=\mathtt{aaababaabbabb}$.}\label{fig:example-ltree}
  \end{figure}
\newpage 

\subsection*{Tree-Structured Interpretation of the Algorithm}
    We show on an example an alternative graphical 
    illustration of the behaviour of the query algorithm.
    To answer an $\LCE(a,b)$ query we traverse a sequence of block-pairs that we call here a \emph{working sequence}
    $(A_0,B_0), (A_1,B_1), \ldots$
    Note that the blocks form a tree-like structure, which lets us depict them using a binary tree; see \cref{fig:working}.

  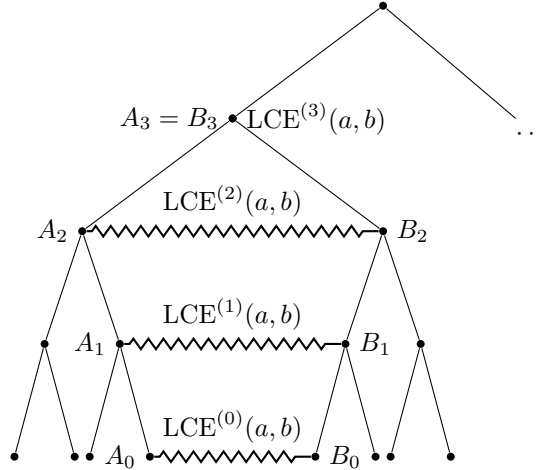
\begin{figure}
    \begin{center}
      \begin{tikzpicture}[level 1/.style={sibling distance=4cm},
level 2/.style={sibling distance=4cm},
level 3/.style={sibling distance=1cm},
level 4/.style={sibling distance=0.8cm}]
\node[internal] (z){}
  child {node [internal, label={left:$A_3=B_3$}] {}
         node [internal, label={right:$\LevelLCE{3}(a,b)$}] {}
    child {node [internal, label={left:$A_2$}] (a2) {}
      child {node[internal]{}
        child{node[internal]{}}
        child{node[internal]{}}
      } 
      child {node [internal, label={left:$A_1$}] (a1) {}
        child {node[internal]{}}
        child {node [internal, label={left:$A_0$}] (a0) {}}  
      }
    }
    child {node [internal, label={right:$B_2$}] (b2) {}
      child {node [internal, label={right:$B_1$}] (b1) {}
        child {node [internal, label={right:$B_0$}] (b0) {}}
        child {node[internal]{}}  
      }
      child {node[internal]{}
        child {node[internal]{}}
        child {node[internal]{}}
      } 
    }
  }
  child {node [below] {$\cdots$}}
;
\path[zig] (a2) --  (b2) node [midway, above] {$\LevelLCE{2}(a,b)$};
\path[zig] (a1) --  (b1) node [midway, above] {$\LevelLCE{1}(a,b)$};
\path[zig] (a0) --  (b0) node [midway, above] {$\LevelLCE{0}(a,b)$};
\end{tikzpicture}
    \end{center}
    \caption{
      A working sequence of block-pairs $(A_0,B_0), (A_1,B_1), \ldots$ used to answer an $\LCE(a,b)$ query.
      Here $A_0 = \{a\}$ and $B_0 = \{b\}$.
    }\label{fig:working}
  \end{figure}

    The behaviour of the algorithm depends on the states of the block-pairs.
    For example, if all the block-pairs are in state $\A$, their state becomes $\B$ and the sequence of queries
    is interrupted at the first level $i$ for which the query is short (i.e., $\LCE(a,b)<3 \cdot 2^i$).
    On the other hand, if a block-pair $(A_i,B_i)$ is in state $\B(a_0,b_0,L)$, then the query may be answered immediately
    if $a_0-a=b_0-b$, and otherwise two additional $\LevelLCE{i+1}$ queries are triggered and the state of the block-pair becomes $\C$.

\end{document}